%% file: gapsparse-full.tex
\documentclass{article}
\input{stdinc.tex}

\title{On the Hardness of Learning Sparse Parities}
\author{Arnab Bhattacharyya\thanks{Department of Computer Science and Automation, Indian Institute of Science, Bangalore, India. Emails: {\tt \{arnabb, ameet.gadekar, suprovat.ghoshal\}@csa.iisc.ernet.in}. AB supported in part by DST Ramanujan Fellowship.}
 \and Ameet Gadekar$^*$
 \and Suprovat Ghoshal$^*$ 
 \and Rishi Saket\thanks{IBM Research, Bangalore, India. Email: {\tt rissaket@in.ibm.com}.}}

\newcommand{\keven}{$k$-\textsc{EvenSet}\xspace}
\newcommand{\kvec}{$k$-\textsc{VectorSum}\xspace}

\newcommand{\kclq}{$k$-\textsc{Clique}\xspace}
\newcommand{\tsat}{$3$-\textsc{SAT}\xspace}
\newcommand{\Ft}{\F_2}
\newcommand{\kVS}{{\sc $k$-VectorSum}\xspace}
\newcommand{\GkVS}{{\sc Gap-$k$-VectorSum}\xspace}
\newcommand{\mc}[1]{\ensuremath{\mathcal{#1}}\xspace}

\date{\today}

\begin{document}
\maketitle
\thispagestyle{empty}

\begin{abstract}
This work investigates the hardness of computing sparse solutions to systems of linear equations over $\F_2$. Consider the \keven problem: given a homogeneous system of linear equations over $\F_2$ on $n$ variables, decide if there exists a nonzero solution of Hamming weight at most $k$ (i.e. a $k$-sparse solution). While there is a simple $O(n^{k/2})$-time algorithm for it, establishing fixed parameter intractability for \keven has been a notorious open problem. Towards this goal, we show that unless \kclq can be solved in $n^{o(k)}$ time, \keven has no $\poly(n)\cdot 2^{o(\sqrt{k})}$ time algorithm for all $k$ and no polynomial time algorithm when $k = \omega(\log^2 n)$.

Our work also shows that the non-homogeneous generalization of the problem -- which we call \kvec\ -- is W[1]-hard on instances where the number of equations is $O(k\log n)$, improving on previous reductions which produced $\Omega(n)$ equations. We use the hardness of \kvec as a starting point to prove the result for \keven, and additionally strengthen the former to show the hardness of \emph{approximately learning} \emph{$k$-juntas}. In particular, we prove that given a system of $O(\tn{exp}(O(k))\cdot \log n)$ linear equations, it is W[1]-hard to decide if there is a $k$-sparse linear form satisfying all the equations or any function on at most $k$-variables (a $k$-junta) satisfies at most $(1/2 + \eps)$-fraction of the equations, for any constant $\eps > 0$. In the setting of computational learning, this shows hardness of approximate \emph{non-proper} learning of $k$-parities.
	In a similar vein, we use the hardness of \keven to show that that for any constant $d$, unless \kclq can be solved in $n^{o(k)}$ time, there is no $\poly(m, n)\cdot 2^{o(\sqrt{k})}$ time algorithm to decide whether a given set of $m$ points in $\F_2^n$ satisfies: (i) there exists a non-trivial $k$-sparse homogeneous linear form evaluating to $0$ on all the points, or (ii) any non-trivial degree $d$ polynomial $P$ supported on at most $k$ variables evaluates to zero on $\approx \Pr_{\F_2^n}[P({\bf z}) = 0]$ fraction of the points i.e., $P$ is \emph{fooled} by the set of points.

Lastly, we study the approximation in the sparsity of the solution. Let the \GkVS problem be: given an instance of \kvec of size $n$, decide if there exist a $k$-sparse solution, or every solution is of sparsity at least $k' = (1+\delta_0)k$. Assuming ETH, we show that for some constants $c_0$, $\delta_0 > 0$ there is no $\tn{poly}(n)$ time algorithm for \GkVS when $k = \omega((\log \log n)^{c_0})$.


\end{abstract}

\newpage
\setcounter{page}{1}

\input{intro}

\input{prelim}

\input{kclique}

\input{evenset}

\input{k-parity-stdinc}

\input{evenset-polynomials}

\input{Gap-k-parity-stdinc}

\section*{Acknowledgements}
We thank Subhash Khot for his permission to include his proof for the hardness of learning parities using polynomials. Also, thanks to Ryan Williams for encouraging and stimulating conversations and to Rocco Servedio for bringing \cite{AKL09} to our attention	.

\bibliographystyle{alpha}
\bibliography{sparse}

\appendix

\input{k-parity-algo}

\input{Parity-stdinc}
\end{document}

%% file: stdinc.tex
\usepackage{xspace}
\usepackage{verbatim}
\usepackage{color}
\usepackage{graphicx}
\usepackage{tabularx}

\usepackage{amsmath}
\usepackage[showonlyrefs]{mathtools}
\usepackage{mathstyle}
\usepackage{empheq}
\usepackage{amstext,amssymb,amsfonts}
\usepackage{fullpage}
\usepackage{nicefrac}
\usepackage{bm}
\usepackage{bbm}

\usepackage{todonotes}

\newcounter{mynotes}
\usepackage{textcomp,setspace}

\usepackage{nameref}
\usepackage[pagebackref,colorlinks,linkcolor=blue,filecolor = blue,citecolor = blue, urlcolor = blue, hyperfootnotes=false]{hyperref}
\usepackage{color}
\usepackage[capitalize]{cleveref}
\renewcommand{\cref}{\Cref}

\usepackage{amsthm}
\usepackage{thmtools}

\declaretheorem[within=section]{theorem}
\declaretheorem[sibling=theorem]{corollary}

\declaretheorem[sibling=theorem]{lemma}

\declaretheorem[sibling=theorem]{definition}

\declaretheorem[sibling=theorem]{proposition}

\newcounter{termcounter}
\renewcommand{\thetermcounter}{\Alph{termcounter}}
\crefname{term}{term}{terms}
\creflabelformat{term}{\textup{(#2#1#3)}}
\renewcommand{\cref}{\Cref}

\makeatletter
\def\term{\@ifnextchar[\term@optarg\term@noarg}
\def\term@optarg[#1]#2{%
  \textup{(#1)}%
  \def\@currentlabel{#1}%
  \def\cref@currentlabel{[][2147483647][]#1}%
  \cref@label[term]{#2}}
\def\term@noarg#1{%
  \refstepcounter{termcounter}%
  \textup{(\thetermcounter)}%
  \cref@label[term]{#1}}
\makeatother

\newcommand{\ignore}[1]{}

\newcommand{\poly}{\mathrm{poly}}

\definecolor{DSred}{rgb}{1,0,0}

\renewcommand{\leq}{\leqslant}
\renewcommand{\geq}{\geqslant}
\renewcommand{\ge}{\geqslant}
\renewcommand{\le}{\leqslant}
\renewcommand{\epsilon}{\varepsilon}
\newcommand{\eps}{\epsilon}

\newcommand{\F}{\mathbb{F}}

\newcommand{\cC}{\mathcal C}
\newcommand{\cD}{\mathcal D}

\newcommand{\cF}{\mathcal F}

\newcommand{\cU}{\mathcal U}

\newcommand{\cW}{\mathcal W}

\newcommand{\cY}{\mathcal Y}

\newcommand{\tn}[1]{\ensuremath{\textnormal{#1}}\xspace}

\newcommand{\Esymb}{{\bf E}}

\newcommand{\Psymb}{{\bf Pr}}

\DeclareMathOperator*{\E}{\Esymb}

\DeclareMathOperator*{\ProbOp}{\Psymb}

\renewcommand{\Pr}{\ProbOp}

\newcommand{\ComplexityFont}[1]{\ensuremath{\mathsf{#1}}}

\newcommand{\NP}{\ComplexityFont{NP}}

\newcommand{\RP}{\ComplexityFont{RP}}


%% file: intro.tex
\section{Introduction}

Given a system of linear equations over $\F_2$, does there exist a sparse non-trivial solution? This question is studied in different guises in several areas of mathematics and computer science. For instance, in coding theory, if the system of linear equations is ${\bf Mx} = {\bf 0}$ where ${\bf M}$ is the parity check matrix of a binary code, then the minimum (Hamming) weight of a nonzero solution is the distance of the code. This also captures the problem of determining whether a binary matroid has a short cycle, as the latter reduces to deciding whether there is a sparse nonzero ${\bf x}$ such that ${\bf Mx} = {\bf 0}$.
In learning theory, the well known sparse parity problem is: given a binary matrix ${\bf M}$ and a vector ${\bf b}$ decide whether there is a small weight nonzero vector ${\bf x}$ satisfying ${\bf Mx} = {\bf b}$. The version where ${\bf Mx}$ is required to equal ${\bf b}$ in most coordinates, but not necessarily all, is also well studied as the problem of learning noisy parities. 

Let a vector ${\bf x} \in \F_2^n$ be called {\em $k$-sparse} if it is nonzero in at most $k$ positions, i.e. it has Hamming weight at most $k$. 
In this work, we show that learning a $k$-sparse solution to a system of linear equations is fixed parameter intractable, even when (i) the number of equations is only \emph{logarithmic} in the number of variables, (ii) the learning is allowed to be \emph{approximate}, i.e. satisfy  only $51\%$ of the equations and, (iii) is allowed to output as hypothesis any function (junta) supported on at most $k$ variables. 
We also prove variants of these results for the case when the system of equations is homogeneous, which correspond to hardness of the well known \keven problem.
Note that it is always possible to recover a $k$-sparse solution in $O(n^k)$ time simply by enumerating over all $k$-sparse vectors. Our results show that for many settings of $k$, no substantially faster algorithm is possible for \keven unless widely believed conjectures are false. Assuming similar conjectures, we also rule out fast algorithms for
learning $\gamma k$-sparse solutions to a linear system promising the existence of a $k$ sparse solutions, for some $\gamma > 1$.

In the next few paragraphs we recall previous related work and place our results in their context. Let us first formally define the basic objects of our study:

\begin{definition} \kvec: Given a matrix ${\bf M} \in \F_2^{m \times n}$ and a vector ${\bf b} \in \F_2^m$, and a positive integer $k$ as parameter, decide if there exists a $k$-sparse vector ${\bf x}$ such that ${\bf Mx} = {\bf b}$.
\end{definition}

\begin{definition} \keven:
Given a matrix ${\bf M} \in \F_2^{m \times n}$, and a positive integer $k$ as parameter, decide if there exists a $k$-sparse vector ${\bf x}$ such that ${\bf Mx} = {\bf 0}$.
\end{definition}

\noindent {\em Remark}:
In the language of coding theory, \kvec is also known as the \textsc{MaximumLikelihoodDecoding} problem and \keven as the \textsc{MinimumDistance} problem. 

\medskip
Clearly, \kvec is as hard as \keven\footnote{The name \keven is from the following interpretation of the problem: given a set system $\cF$ over a universe $U$ and a parameter $k$, find a nonempty subset $S \subseteq U$ of size at most $k$ such that the intersection of $S$ with every set in $\cF$ has even size.}. The \kvec problem was shown to be W[1]-hard\footnote{Standard definitions in parameterized complexity appear in \cref{sec:prelim}.} by Downey, Fellows, Vardy and Whittle \cite{DFVW99}, even in the special case of the vector ${\bf b}$ consisting of all $1$'s.  More recently, Bhattacharyya, Indyk, Woodruff and Xie \cite{BIWX11} showed that the time complexity of \kvec is $\min(2^{\Theta(m)}, n^{\Theta(k)})$, assuming \tsat has no $2^{o(n)}$ time algorithm. 

 In contrast, the complexity of \keven remains unresolved, other than its containment in W[2] shown in \cite{DFVW99}. Proving W[1]-hardness for \keven was listed as an open problem in Downey and Fellows' 1999 monograph \cite{DF99} and has been reiterated more recently in lists of open problems \cite{FM12, FGMS12}. Note that if we ask for a vector ${\bf x}$ whose weight is \emph{exactly} $k$ instead of at most $k$, the problem is known to be W[1]-hard \cite{DFVW99}. Our work gives evidence ruling out efficient algorithms for \keven for a wide range of settings of $k$.
 
 In the non-parameterized setting, where $k$ is  part of the input, these problems are very well-studied. Vardy showed that \textsc{EvenSet} (or \textsc{MinimumDistance}) is \NP-hard \cite{Var97}.  The question of approximating $k$, the minimum distance of the associated code, has also received attention. Dumer, Micciancio and Sudan \cite{DMS03} showed that if $\RP\neq\NP$, then $k$ is hard to approximate within some constant factor $\gamma > 1$. Their reduction was derandomized by Cheng and Wan \cite{CW08, CW09}, and subsequently Austrin and Khot \cite{AK14} gave a simpler deterministic reduction for this problem. The results of \cite{CW08, CW09} and \cite{AK14} were further strengthened by Micciancio \cite{Mic14}.
 
From a computational learning perspective, the \kvec problem can be restated as: given an $m$-sized set of $n$-dimensional point and value pairs over $\F_2$, decide if there exists a parity supported on at most $k$ variables (i.e. a $k$-parity) that is consistent with all the pairs. This has been extensively studied as a promise problem when the points are uniformly generated. Note that in this case, if $m = \Omega(n)$, there is a unique solution w.h.p and can be found efficiently by Gaussian elimination. On the other hand, for $m = O(k \log n)$, the best known running time  of $O(n^{k/2})$ is given in \cite{KS06} (credited to Dan Spielman). Obtaining a polynomial time algorithm for $m = \poly(k \log n)$ would imply {\em attribute-efficient learning} of $k$-parities and is a long-standing open problem in the area \cite{Blum96}. The best known dependence between $m$ and the running time for this problem is described in \cite{BGM10, BGR15}.  Our work proves the hardness of \kvec when $m = O(k\log n)$, showing evidence which rules out polynomial time algorithms for learning $k$-parity when the input is generated adversarially.

A natural question studied in this work is  whether one can do better if the learning algorithm is allowed to be \emph{non-proper} (i.e., output a hypothesis that is not a $k$-parity) and is allowed to not satisfy all the point-value pairs. To further motivate this problem, let us look at the case when $k$ is not fixed. In the absence of noise, Gaussian elimination can efficiently recover a consistent parity. The harder case is the agnostic (noisy) setting which promises that there is a parity consistent with at least $1 - \eps$ fraction of the point-value pairs. When the points are generated uniformly at random,  one can learn the parity in time  $2^{O(n/\log n)}$ \cite{FGKP09, BKW03}. 

On the other hand, when the points are adversarially drawn, there is a non-proper algorithm due to Kalai, Mansour and Verbin \cite{KMV08} that runs in time $2^{O(n/\log n)}$ and outputs a circuit $C$ which is consistent with at least $\left(1 - \eps - 2^{-n^{0.99}}\right)$ of the point value pairs. H\aa stad's inapproximability for {\sc Max-$3$LIN}~\cite{Has01} implies that learning a noisy parity in the adversarial setting is \NP-hard, even for $1/2 + \eps$ accuracy, for any constant $\eps > 0$. Gopalan, Khot and Saket~\cite{GKS10} showed that achieving an accuracy of $1 - 1/2^d + \eps$ using degree-$d$ polynomials as hypotheses is \NP-hard. Subsequently, Khot~\cite{Khot-personal} proved an optimal bound of $1/2 + \eps$ for learning by constant degree polynomials\footnote{As far as we know, this result is unpublished although it was communicated to the fourth author of this paper. We include with his permission a proof of Khot's result to illustrate some of the techniques which inspire part of this work.}. Our work studies the intractability of approximate non-proper learning of $k$-parity, and extends the hardness result for \kvec to learning by juntas of $k$ variables, and for \keven to learning using constant degree polynomials on $k$ variables.

Another interesting question in the parameterized setting is related to a gap in the sparsity parameter $k$, i.e. how tractable it is to learn a $\gamma k$-sparse solution when the existence of a $k$-sparse solution is guaranteed, for some constant $\gamma > 1$ (or $\gamma < 1$ in case of a maximization problem). Previously, Bonnet et al.~\cite{BEKP15} and Khot and Shinkar~\cite{KS15} studied this problem for \kclq, and both these works show conditional hardness results. In our work we prove a ``gap in $k$'' hardness result for \kvec similar to that obtained in \cite{BEKP15} for \kclq. 

In the rest of this section we formally describe our results for \kvec and \keven, and give a brief description of the techniques used to obtain them.

\subsection{Our Results}

\subsubsection*{Hardness of exact problems}
We begin by giving a reduction from \kclq showing the W$[1]$-hardness of \kvec on instances which have a small number of rows.

\begin{theorem}[{W[1]-hardness of \kvec}]
\label{thm:vechard}
The \kvec problem is \tn{W$[1]$}-hard on instances $({\bf M}, {\bf b})$ where ${\bf M} \in \F_2^{m \times n}$ and ${\bf b} \in \F_2^m$ such that $m = O(k \log n)$. Our reduction implies, in particular, that \kvec does not admit an $n^{o(\sqrt{k})}$ time algorithm on such instances, unless \kclq on $r$-vertex graphs has an $r^{o(k)}$ time algorithm.
\end{theorem}

As far as we know, in previous proofs of the W[1]-hardness of \kvec \cite{DFVW99, CFKLMPPS15}, the number of rows in the matrix output by the reduction was linear in $n$. Our proof is inspired by a recent proof of the W[1]-hardness of $k$-\textsc{Sum} \cite{ALW14}. 
Also, in Appendix \ref{sec-kparityalgo}, we give a simple $O(n\cdot 2^m)$ time algorithm for \kvec, which suggests that  $m$ cannot be made sublogarithmic in $n$ for hard instances.



Next, we give a hardness reduction from \kvec  to the \keven problem.

\begin{theorem}[Hardness of \keven]\label{thm:even}
There is an \tn{FPT} reduction from an instance $({\bf M}, {\bf b})$ of \kvec, where ${\bf M} \in \F_2^{m \times n}$ and ${\bf b} \in \F_2^m$,  to an instance ${\bf M}'$ of $O((k \log n)^2)$-\textsc{EvenSet}, where ${\bf M}' \in \F_2^{m' \times n'}$ such that both $m'$ and $n'$ are bounded by fixed polynomials in $m$ and $n$.
\end{theorem}

Using \cref{thm:vechard}, the above yields the following corollary.
\begin{corollary}\label{cor:evenlb}
	There does not exist a $\poly(n)$ time algorithm for \keven when $k = \omega(\log^2 n)$, assuming that \kclq does not have a polynomial time algorithm for any $k = \omega(1)$. More generally, under the same assumption,  \keven does not admit a $\poly(n) \cdot 2^{o(\sqrt{k})}$ time algorithm for unrestricted $k$.
\end{corollary}
\begin{proof}
	Suppose there is a $T(n,k)$ algorithm for \keven. Chaining together the reductions in \cref{thm:vechard} and \cref{thm:even}, we get a $T(\poly(n), k^4 \log^2 n)$ algorithm for \kclq. Choosing $k = \omega(1)$ implies the first part of the corollary. For the second part, observe that if $f(x) = 2^{o(\sqrt{x})}$, then $f(k^4 \log^2 n) = n^{o(1)}$ for some $k = \omega_n(1)$. 
\end{proof}

To the best of our knowledge, \cref{cor:evenlb} gives the first nontrivial hardness results for parameterized \keven. \cref{thm:even} is obtained by adapting the hardness reduction 
for the inapproximability of \textsc{MinimumDistance} by Austrin and Khot \cite{AK14} to the parameterized setting.

\subsubsection*{Hardness of non-proper and approximately learning sparse parities}
The hardness for \kvec proved in \cref{thm:vechard} can be restated in terms of W[1]-hardness of learning $k$-parity, i.e., linear forms depending on at most $k$ variables\footnote{Note that Theorem \ref{thm:vechard} as stated shows hardness of learning homogeneous $k$-parity i.e., homogeneous $k$-sparse linear forms (without the constant term). The result can easily be made to hold for any general $k$-parities by adding a point-value pair which is $({\bf 0}, 0)$.}.

\begin{theorem}[\cref{thm:vechard} restated]\label{thm:linhard}
The following is \tn{W[1]}-hard: given $m = O(k \log n)$ point-value pairs $\{({\bf y}_i, a_i)\}_{i=1}^m \subseteq \F_2^n \times \F_2$, decide whether there exists a $k$-parity $L$ which satisfies all the point-value pairs, i.e., $L({\bf y}_i) = a_i$ for all $i = 1, \dots, m$. 
\end{theorem}


Next, we strengthen the above theorem in two ways. We show that the W[1]-hardness holds for learning a $k$-parity using a $k$-junta, and additionally 
 for any desired accuracy exceeding 50\%. 
Here, a $k$-junta is any function depending on at most $k$ variables.  

\begin{theorem}\label{thm-kparityjunta} The following is W$[1]$-hard:
for any constant $\delta > 0$, given
$m = O(k\cdot 2^{3k}\cdot (\log n)/\delta^3)$ point-value pairs $\{({\bf z}_i,
b_i)\}_{i=1}^m
\subseteq \F_2^n\times \F_2$,
decide whether: 

\smallskip
\noindent
\tn{YES Case:} There exists a $k$-parity which satisfies all the point-value pairs.

\smallskip
\noindent
\tn{NO Case.} Any function $f : \F_2^n \mapsto \F_2$ depending on at most $k$ variables satisfies at most $1/2 + \delta$ fraction of the point value pairs.
\end{theorem}

\cref{thm-kparityjunta} also implies hardness for approximately learning $k$-juntas, in comparison to previous W$[2]$-hardness of \emph{exactly} learning $k$-juntas shown by Arvind, K{\"{o}}bler and Lindner~\cite{AKL09}.  Note that the current best algorithm for learning $k$-junta, even over the uniform distribution, takes $n^{\Omega(k)}$ time \cite{Val12, MOS03}.

We similarly strengthen Corollary \ref{cor:evenlb} to rule out efficient algorithms for approximately learning a $k$-sparse solution to a homogeneous linear system using constant degree polynomials supported on at most $k$ variables.
\begin{theorem}\label{thm:kevenwithpolys}
	Assume that \kclq does not have a $\tn{poly}(n)$ time algorithm for any $k = \omega(1)$. Then for any constant $\delta > 0$ and positive integer $d$, there is no  $\poly(m,n) \cdot 2^{o(\sqrt{k})}$ time algorithm to
decide whether a given set of 
$m$ points $\{{\bf z}_i\}_{i=1}^m
\subseteq \F_2^n$ satisfy:

\smallskip
\noindent
\tn{YES Case:} There exists a nonzero $k$-parity $L$ such that $L({\bf z}_i) = 0$ for all $i=1,\dots, m$.

\smallskip
\noindent
\tn{NO Case.} Any non-trivial degree $d$ polynomial $P : \F_2^n \mapsto \F_2$ depending on at most $k$ variables satisfies $P({\bf z}_i) = 0$ for
at most $\left(\Pr_{{\bf z} \in \F_2^n}[P({\bf z}) = 0] + \delta\right)$ fraction of the points. 
\end{theorem}
The proof of the above theorem relies on an application of Viola's~\cite{Viola09} pseudorandom generator for constant degree polynomials, and is inspired by Khot's~\cite{Khot-personal}  \NP-hardness of learning linear forms using constant degree polynomials.

\subsubsection*{Gap in sparsity parameter}
Using techniques similar to those employed in \cite{BEKP15}, we prove the following \emph{gap in $k$} hardness for \kvec, i.e., hardness of \GkVS.
\begin{theorem}\label{thm-gapink} Assuming the Exponential Time
Hypothesis, there are universal constants $\delta_0 >0$ and $c_0$ 
such that  there is no $\tn{poly}(N)$ time algorithm to determine whether an
instance of \GkVS of  size $N$ admits a solution of sparsity $k$ or all solutions
are of sparsity at least $(1 + \delta_0)k$, 
for any $k = \omega((\log\log N)^{c_0})$. More generally, under the same assumption, this problem does not admit an $N^{O(k/\omega((\log\log N)^{c_0}))}$ time algorithm for unrestricted $k$.
\end{theorem}

\subsection{Our Techniques}

Our first result, Theorem \ref{thm:vechard}, is based on a gadget reduction from an $n$-vertex instance of \kclq creating columns of ${\bf M}$ corresponding to the vertices and edges of the graph along with a target vector ${\bf b}$. Unlike previous reductions which used dimension linear in the number of vertices, we reuse the same set of  coordinates for the vertices and edges by assigning unique logarithmic length patterns to each vertex. In total we create $k$ columns for each vertex and ${k \choose 2}$ columns for each edge, using $O(k^2\log n)$ coordinates. The target vector ${\bf b}$ ensures that a solution always has at least $k$ + ${k \choose 2}$ columns, which suffices in the YES case while the NO case requires strictly more columns to sum to ${\bf b}$. 

For proving Theorem \ref{thm:even}, we homogenize the instance of Theorem \ref{thm:vechard} by including ${\bf b}$ as a column of ${\bf M}$. To force the solution to always choose ${\bf b}$ we use the approach of Austrin and Khot~\cite{AK14} who face the same issue when reducing to the \textsc{MinimumDistance} problem. Since we need to retain the bound on the sparsity of the solution, we cannot use their techniques directly. Instead, we construct a small length sketch of a purported sparse solution and use it as an input to a tensor code based amplification gadget used in \cite{AK14}. Our construction however inflates the parameter $k$ to $O((k\log n)^2)$.

The hardness of approximately learning $k$-parities with $k$-juntas given in Theorem \ref{thm-kparityjunta} is obtained by transforming the instance of Theorem \ref{thm:vechard} using an $\eps$-balanced code, along with an analysis of the Fourier spectrum of any $k$-junta on the resulting distribution. In contrast, Theorem \ref{thm:kevenwithpolys} is obtained by using the instance of Theorem \ref{thm:even} (appropriately transformed using an $\eps$-balanced code) as an input to Viola's construction~\cite{Viola09} of pseudorandom generators for degree $d$ polynomials. Note that the $\tn{exp}(k)$ blowup in the reduction for Theorem \ref{thm-kparityjunta} rules out its use for proving Theorem \ref{thm:kevenwithpolys} due to the presence of a $(\log ^2 n)$ factor in the sparsity parameter of the instance obtained in Theorem \ref{thm:even}. On the other hand, the non-homogeneity of the \kvec problem hinders the use of Viola's pseudorandom generator for proving a version (for degree $d$ polynomials on $k$ variables instead of $k$-juntas) of Theorem \ref{thm-kparityjunta} which avoids the $\tn{exp}(k)$ blowup.

For Theorem \ref{thm-gapink}, we use the improved \emph{sparsification lemma} of Calabro, Impagliazzo and Paturi~\cite{CIP06} followed by Dinur's almost linear PCP construction~\cite{Dinur07} to reduce an $n$-variable {\sc $3$-SAT} instance to $2^{\eps n}$ {\sc Gap-$3$-SAT} instances with almost linear in $n$ clauses and variables. For each instance a corresponding \kvec instance is created by partitioning the clauses into $k$ blocks and adding $\F_2$-valued variables for partial assignments to each block along with non-triviality and consistency equations. In the YES case setting one variable from each block to $1$ (i.e. a $k$-sparse solution) suffices, whereas in the NO case at least $\gamma k$ variables need to be set to $1$, for some constant $\gamma > 1$. The parameters are such that an efficient algorithm to decide the YES and NO cases would violate the Exponential Time Hypothesis for {\sc $3$-SAT}.

\medskip
\noindent
{\bf Organization of the paper.} Theorem \ref{thm:vechard} is proved in Section \ref{sec:vechard}, and using it as the starting point Theorem \ref{thm:even} is proved in Section \ref{sec:even}. The reduction proving Theorem \ref{thm-kparityjunta} is given in  \ref{sec:kparityjunta}, and starts with a restatement of Theorem \ref{thm:vechard}. The proofs of Theorems \ref{thm:kevenwithpolys} and \ref{thm-gapink}  are provided in Section \ref{sec:kevenwithpolys} and Section \ref{sec-gapink} respectively.

We also include in Appendix \ref{sec-Khot-personal} a proof of Khot's~\cite{Khot-personal} result on NP-hardness of approximately learning linear forms using constant degree polynomials to illustrate the use of Viola's pseudorandom generator~\cite{Viola09} which is also used in the proof of Theorem \ref{thm:kevenwithpolys}.

In the next section we give some definitions and results which shall prove useful for the subsequent proofs.

%% file: prelim.tex
\section{Preliminaries}\label{sec:prelim}

\subsection{Parameterized Complexity}
A {\em parameterization} of a problem is a $\tn{poly}(n)$-time computable function that assigns an integer $k\geq 0$ to each problem instance $x$ of length $n$ (bits). The pair $(x, k)$ is an instance of the corresponding parameterized problem. The parameterized problem is said to be {\em fixed parameter tractable (FPT)} if it admits an algorithm that runs in time $f(k) \cdot \poly(n)$ where $k$ is the parameter of the input, $n$ is the size of the input, and $f$ is an arbitrary computable function. The {\em \tn{W}-hierarchy}, introduced by Downey and Fellows~\cite{DF95,DF99}, is a sequence of parameterized complexity classes with $\tn{FPT} = \tn{W[0]} \subseteq \tn{W[1]} \subseteq \tn{W[2]} \subseteq \cdots$. It is widely believed that FPT $\neq$ W[1].

These hierarchical classes admit notions of completeness and hardness under FPT reductions i.e., $f(k) \cdot \poly(n)$-time transformations from a problem $A$ instance $(x, k)$ where $|x| = n$, to an instance $(x', k')$ of problem $B$ where $|x'| = \tn{poly}(n)$ and $k'$ is bounded by $f(k)$. For example, consider the \kclq problem: given a graph $G$ on $n$ vertices and an integer parameter $k$, decide if $G$ has a clique of size $k$. The  \kclq problem is W[1]-complete, and serves as a canonical hard problem for 
many W[1]-hardness reductions including those in this work.

For a precise definition of the W-hierarchy, and a general background on parameterized algorithms and complexity, see \cite{DF99, FG06, CFKLMPPS15}.

%
%

\subsection{Coding Theoretic Tools}

Our hardness reductions use some basic results from coding theory. 
For our purposes, we shall be restricting our attention to \emph{linear} codes over $\F_2$ i.e., those which form linear subspaces. A code $\mathcal{C} \subseteq  \F^n_2$ is said to be a $[n,k,d]$-binary linear code if $\mathcal{C}$ forms a $k$-dimensional subspace of $\F^n_2$ such that all nonzero elements (codewords) in $\mathcal{C}$ are of Hamming weight at least $d$. We use weight $\tn{wt}({\bf x})$ of a codeword ${\bf x}$ to denote its Hamming weight, \emph{distance} of a code to denote the minimum weight of any nonzero codeword, and \emph{rate} to denote the fraction $k/n$. A \emph{generator} matrix ${\bf G} \in \F_2^{n\times k}$ for $\mathcal{C}$ is such that $\mathcal{C} = \{{\bf Gx}\,\mid\, {\bf x} \in \F_2^k\}$. 
Also associated with $\mathcal{C}$ is a \emph{parity check matrix} ${\bf G}^{\perp} \in \F^{(n - k) \times n}_2$ satisfying: ${\bf G}^{\perp}{\bf y} = {\bf 0 }$ \emph{iff} ${\bf y} \in \mathcal{C}$. 
We shall use the generator and parity check matrices of well studied code constructions
whose properties we state below.  
\begin{theorem}[BCH Codes, Theorem 3 ~\cite{BC60}]
	The dimension of the BCH code of block length $n = (2^m - 1)$ and distance $d$ is at least $\left(n - \lceil \frac{d-1}{2} \rceil m\right)$.
\end{theorem}
While the above theorem restricts the block length to be of the form $(2^m - 1)$,  for general $n$ we can use as the parity check matrix any $n$ columns of the parity check matrix of a BCH code of the minimum length $(2^m - 1)$ greater than or equal to $n$. In particular, we have the following corollary tailored for our purpose.
\begin{corollary}\label{eq:bch_dist}
For all lengths $n$ and positive integers $k < n$,
there exists a parity check matrix ${\bf R} \in \F^{20k\log n \times n}_2$ such that ${\bf R}{\bf x} \ne 0$ whenever $0 < \tn{wt}({\bf x}) < 18k$. Moreover, this matrix can be computed in time $\tn{poly}(n,k)$.
\end{corollary}
The following explicit family of $\eps$-balanced binary linear codes of constant rate was given by Alon et al.~\cite{ABNNR92}.
%


\begin{theorem}[$\eps$-balanced codes~\cite{ABNNR92}] \label{thm:balanced}
	There exists an explicit family of codes $\mathcal{C} \subseteq \F_2^n$ such that every codeword in $\mathcal{C}$ has normalized weight in the range $\left[1/2 - \epsilon, 1/2 + \eps\right]$, and rate $\Omega(\epsilon^3)$, which can be constructed in time $\tn{poly}(n, \frac{1}{\epsilon})$, where $\epsilon > 0$ is an arbitrarily small constant. 
\end{theorem}

Given a linear code $\mathcal{C} \subseteq \F^n_2$, the product code $\mathcal{C}^{\otimes2}$ consists of $n \times n$ matrices where each row and each column belongs to $\cC$; equivalently, $\cC^{\otimes 2} = \{{\bf GXG}^{\mathrm{T}} :  {\bf X} \in \F_2^{k \times k}\}$ where ${\bf G} \in \F_2^{n \times k}$ is the generator matrix for the code $\cC$. If the distance $d(\mathcal{C}) = d$, then it is easy to verify that $d(\mathcal{C}^{\otimes 2}) \ge d^2$. However, we shall use the following lemma from \cite{AK14} for a tighter lower bound on the Hamming weight when the code word satisfies certain properties.
\begin{lemma}[Density of Product Codes~\cite{AK14}]\label{thm:prod_codes}
	Let $\mathcal{C} \subseteq {\F}^n_2$ be a binary linear code of distance $d = d(\mathcal{C})$, and let ${\bf Y} \in \mathcal{C}^{\otimes 2}$ be a nonzero codeword with the additional properties that ${\rm diag}({\bf Y}) = {\bf 0}$, and ${\bf Y} = {\bf Y}^{\sf T}$. Then, the Hamming weight of ${\bf Y}$ is at least $\frac{3}{2}d^2$.
\end{lemma}

\input{prelim-extra}

%% file: prelim-extra.tex
\subsection{Some Useful Tools}\label{sec:prelim-appendix}
The proof of Theorem \ref{thm:kevenwithpolys} in Section 
\ref{sec:kevenwithpolys} and Khot's~\cite{Khot-personal} result 
given in Appendix \ref{sec-Khot-personal} use Viola's~\cite{Viola09}
construction of pseudorandom generators which we describe below.
\begin{definition}\label{defn:fooling}
	A distribution $\cD$ over $\F_2^n$ is said to
	\emph{$\eps$-fool}
	degree $d$ polynomials in $n$-variables over $\F_2$ if for any
	degree $d$ polynomial $P$:
	$$\left|\E_{{\bf z} \leftarrow \cD}\left[e(P({\bf z}))\right]
	- \E_{{\bf z} \leftarrow \cU}\left[e(P({\bf
		z}))\right]\right| \leq \eps,$$
	where $\cU$ is the uniform distribution over $\F_2^n$ and
	$e(x) := (-1)^x$ for $x \in \{0,1\}$.
\end{definition}

\begin{theorem}\label{thm:viola}
	Let $\cY_1, \dots, \cY_d$ be $d$ independent distributions
	on $\F_2^n$
	that each $\eps$-fool linear polynomials. 
	Then the distribution $\cW = \cY_1 + \dots + \cY_d$
	$\eps_d$-fools degree-$d$ polynomials where $\eps_d := 16\cdot
	\eps^{1/2^{d-1}}$.
\end{theorem}
Our proof of \cref{thm-gapink} in Section \ref{sec-gapink} 
uses the following improved
\emph{sparsification} lemma of Calabro, Impagliazzo, and Paturi~\cite{CIP06}.
\begin{lemma}\label{lem-sparsification} There is a deterministic
	algorithm which, for any $\eps > 0$, transforms an $n$-variable
	\tn{$3$-CNF} formula
	$F$ to $F_1, \dots, F_s \in 3\tn{-CNF}$, each on at most $n$ variables 
	s.t.
	\begin{itemize}
		\item[1.] $s\leq 2^{\eps n}$.
		\item[2.] $F$ is satisfiable if and only if at least one of $F_1,
		\dots, F_s$ is satisfiable.
		\item[3.] The number of clauses in each $F_1, \dots, F_s$ is at most 
		$O((1/\eps)^9 \cdot n)$.
		\item[4.] The algorithm runs in time $2^{\eps n}\cdot \tn{poly}(n)$,
		where the degree of the polynomial may depend on $\eps$.
	\end{itemize}
\end{lemma}
We shall also use in Section \ref{sec-gapink} 
the following reduction to Gap-$3$-SAT implied by
the  construction of almost linear sized
PCPs given by Dinur~\cite{Dinur07}.
\begin{theorem}\label{thm-nearlinearPCP} There exist universal constants
	$\gamma_0 > 0 $ and $c_0$, and a polynomial time
	reduction from a $3$\tn{-CNF} formula $F$ on $m$ clauses to a
	$3$\tn{-CNF} formula $F'$ on at most $m(\log m)^{c_0}$ clauses such
	that: (i) \tn{(YES Case)} if $F$ is satisfiable then $F'$ is satisfiable, 
	and (ii) \tn{(NO Case)} if $F$ is unsatisfiable then at most
	$(1-\gamma_0)$ fraction of the clauses of $F'$ are satisfied by any
	assignment.
\end{theorem}

%% file: kclique.tex
\section{W[1]-hardness of \kvec on $O(k\log n)$ Equations}\label{sec:vechard}
The following theorem implies Theorem \ref{thm:vechard}.
\begin{theorem}
	There is an FPT reduction from an instance $G(V,E)$ of \kclq, over $n$ vertices and $m$ edges,  to an instance $({\bf M},{\bf b})$ of $k'$-\textsc{VectorSum}, where ${\bf M} \in {\F}^{d \times n^\prime}_2$ such that $k' = O(k^2)$, $d = O(k^2\log n)$ and $n^\prime$ is polynomial  in $n$ and $k$.	
\end{theorem}
The rest of this section is devoted to proving the above theorem.	
We start by observing that a $k$-clique in a graph $G(V,E)$ can be certified by the pair of mappings $f:[k]\mapsto V$ and $g: {[k] \choose 2} \mapsto E$ , such that $g(i,j) = (f(i),f(j)) \in E \quad \forall i,j \in [k], i < j$. Here, we use ${[k] \choose 2}$ to represent $\{(i,j) \mid 1\leq i < j \leq k\}$.  The underlying idea behind the reduction is to construct ${\bf M}$ and ${\bf b}$ such that $f$ and $g$ exist \emph{iff} there is a sparse set of columns of ${\bf M}$ that sums up to ${\bf b}$.
	

\medskip
\noindent
{\bf Construction of ${\bf M}$ and ${\bf b}$.}
	Let $G(V,E)$ be a \kclq instance on $n  = |V|$ vertices and $m = |E|$ edges, where $V = \{v_1,v_2,\ldots,v_n\}$.	
	For each vertex $v_i \in V$, assign a distinct $N = \lceil \log (n+1) \rceil$ bit nonzero binary pattern denoted by ${\bf q}_i \in {\F}^N_2$. 
	We first construct a set of vectors -- which shall be the columns of ${\bf M}$ -- corresponding to the vertices and edges. The dimension over which the vectors are defined is partitioned into three sets of coordinates: 
	
	\medskip
	\noindent\emph{Edge-Vertex Incidence Coordinates}: These consist of $k$ slots, where each slot consists of $(k-1)$ subslots, and each subslot in turn consists of $N$ coordinates. In any column of ${\bf M}$, a subslot may either contain the $N$-length pattern of a vertex, or it might be all zeros.

	\medskip
	\noindent\emph{Edge Indicator Coordinates}: These are a set of ${k \choose 2}$ coordinates corresponding to  $\{(i,j) \mid 1\leq i < j \leq k\}$, indicating whether the vector represents an edge mapped from $(i,j)$. Any column of ${\bf M}$ may have at most one of these coordinates set to $1$.

	\medskip
	\noindent\emph{Vertex Indicator Coordinates}: These are a set of $k$ coordinates corresponding to indices $i \in \{1,\dots, k\}$, which indicate whether the vector represents a vertex mapped from $i$. Any column of ${\bf M}$ may have at most one of these coordinates set to $1$.
	
	\medskip
	Thus, each vector is a concatenation of $k(k-1)N$ edge-vertex incidence bits, followed by ${k \choose 2}$ edge indicator bits and $k$ vertex indicator bits,  so that $d = k(k-1)N + {k \choose 2} + k = O(k^2\log n)$. For ease of notation, let $S^j_l$ represent the $N$-sized subset of coordinates belonging to the subslot $l$ of slot $j$ where $j \in [k]$ and $l \in [k-1]$. We define ${\bf q}_i(S^j_l) \in {\F}^d_2$ to be the vector which contains the pattern of vertex $v_i$ in $S^j_l$, and is zero everywhere else. For $1 \leq i < j \leq k$, let  $\bm{\delta}_{i,j} \in {\F}^d_2$ be the vector which has a $1$ at the edge indicator coordinate corresponding to $(i,j)$, and is $0$ everywhere else. Similarly, $\bm{\delta}_i \in {\F}^d_2 $ is the indicator vector which has its $i$th vertex indicator coordinate set to $1$, everything else being $0$. Use these components we construct the vertex and edge vectors as follows.

	\medskip
	\noindent\emph{Vertex Vectors}: For each vertex  $v_i \in V$ and $j \in [k]$, we introduce a vector $\bm{\eta}(v_i,j) \in {\F}^d_2$ which indicates that vertex $v_i$ is mapped from index (slot) $j$ i.e., $f(j) = v_i$. The vector is constructed as follows: populate each of the $(k-1)$ subslots of the $j$th slot with the pattern of vertex $v_i$ (which is ${\bf q}_i$), and set its $j$th vertex indicator coordinate to $1$. Formally,  $\bm{\eta}(v_i,j) := \sum^{k-1}_{l=1} {\bf q}_i(S^j_l) + \bm{\delta}_j$. For each vertex there are $k$ vertex vectors resulting in a total of $nk$ vertex vectors. 
	
	\medskip
	\noindent\emph{Edge Vectors}: For each edge $e = (v_{i_1},v_{i_2}) \in E$ where $i_1 < i_2$, and $1 \le j_1 < j_2 \le k$, we introduce a vector that indicates that the pair of indices (slots) $(j_1,j_2)$ is mapped to $(v_{i_1},v_{i_2})$ i.e., $g(j_1,j_2) = (v_{i_1},v_{i_2})$ .  We construct the vector $\bm{\eta}(e,j_1,j_2) \in {\F}^d_2$ as follows: populate  $S^{j_1}_{j_2 - 1}$ with the pattern of vertex $v_{i_1}$, and $S^{j_2}_{j_1}$ with the pattern of vertex $v_{i_2}$. Additionally, we set the edge indicator coordinate corresponding to $(j_1,j_2)$ to 1. The vector is formally expressed as, $\bm{\eta}(e,j_1,j_2) :=  {\bf q}_{i_1}(S^{j_1}_{j_2-1}) + {\bf q}_{i_2}(S^{j_2}_{j_1}) + \bm{\delta}_{j_1,j_2}$. Intuitively, for the lower ordered vertex $v_{i_1}$, $\bm{\eta}(e,j_1,j_2)$ cancels out the $(j_2 - 1)$th subslot of slot $j_1$, and for the higher ordered vertex $v_{i_2}$, it cancels out the $j_1$th subslot of its $j_2$th slot. 
Note that we are treating $(v_{i_1},v_{i_2})$ as an \emph{unordered} pair since $i_1 < i_2$. Therefore, for each edge $e \in E$, and for each choice of $1 \leq j_1 < j_2 \leq k$, we introduce one edge vector. Hence, there are a total of $m\cdot{k \choose 2}$ edge vectors in the set. 

\medskip
The vertex and edge vectors constructed above constitute the columns of ${\bf M}$. The target vector ${\bf b}$  ensures that (i) every solution must have at least $k$ vertex vectors, and ${k \choose 2}$ edge vectors and (ii) the vectors must cancel each other out in the Edge-Vertex Incidence coordinates. Formally, ${\bf b} = \sum_{i \in [k]} \bm{\delta}_i + \sum_{1 \le i<j < k} \bm{\delta}_{i,j}$.
In other words, all the edge and vertex indicator coordinates of ${\bf b}$ are set to $1$, and everything else to $0$. 
	
\subsection{YES case}
	We show that if $G(V,E)$ has a $k$-Clique, then there exists a set of $k + {k \choose 2}$ columns of ${\bf M}$ that sum to ${\bf b}$.
	Assume that $v_{i_1},v_{i_2}, \ldots ,v_{i_k}$ form a $k$-clique where $i_1 < i_2 < \dots < i_k$. We select $k$ vertex vectors $\{\bm{\eta}(v_{i_j},j)\}_{j \in [k]}$, and ${k \choose 2}$ edge vectors $\{\bm{\eta}(e, j_1,j_2)\, \mid\, e = (v_{i_{j_1}}, v_{i_{j_2}}), 1 \leq j_1 < j_2 \leq k\}$. Since the $k$ vertices form a clique, these vectors always exists. Observe that for any fixed $j \in [k]$, (i) for $\ell = 1, \dots, j-1$, $\bm{\eta}(v_{i_j},j)$ and  $\bm{\eta}(e, \ell ,j)$ have the same pattern ${\bf q}_{i_j}$ in subslot $\ell$ of slot $j$, where $e = (v_{i_\ell}, v_{i_j})$, and (ii)  for $\ell = j+1, \dots, k$, $\bm{\eta}(v_{i_j},j)$ and $\bm{\eta}(e, j ,\ell)$ have the same pattern ${\bf q}_{i_j}$ in subslot $(\ell-1)$ of slot $j$, where $e = (v_{i_j}, v_{i_{\ell}})$. Thus, the $k + {k \choose 2}$ selected vectors sum to zero on all but the vertex and edge indicator coordinates and thus sum up to ${\bf b}$.

	\subsection{NO Case}

	Suppose for a contradiction that $\mathcal{S}$ is a subset of columns of ${\bf M}$ that sum to ${\bf b}$ and that $|\mathcal{S}| \leq k + {k \choose 2}$. 
	
	\begin{proposition}				\label{prop:exact_ct}
		There are exactly $k$ vertex vectors corresponding to indices (slots) $i \in [k]$ in $\mathcal{S}$. Also, there are exactly ${k \choose 2}$ edge vectors, one for each pair $(i,j)$ ($1\leq i < j \leq k$) of slots,  in $\mathcal{S}$.
	\end{proposition}
	\begin{proof}
		This follows from the observation that there are $k + {k \choose 2}$ nonzero indicator coordinates in the target ${\bf b}$, and each (edge or vertex) vector contributes exactly one nonzero (edge or vertex) indicator coordinate. Therefore, by a  counting argument, $k$ vertex vectors, one each for the indices (slots) $i \in [k]$, must contribute to the $k$ vertex indicator bits. Similarly, ${k \choose 2}$ edge vectors, one for each pair of slots $(i,j)$ ($1\leq i < j \leq k$), must contribute to the ${k \choose 2}$ edge indicator bits.
	\end{proof}

	The above proposition implies that for each pair of vertex vectors there is exactly one edge vector which has a common populated subslot with each of them. So there are exactly $(k-1)$ edge vectors which share a populated subslot with any given vertex vector in $\mathcal{S}$. 

Since the $k$ vertex vectors in $\mathcal{S}$ populate distinct slots, in total $k(k-1)$ subslots are populated by the sum of the $k$ vertex vectors. Note that any edge vector populates exactly $2$ subslots. Thus, for the ${k \choose 2} = k(k-1)/2$ edge vectors in $\mathcal{S}$ to sum up to the values in $k(k-1)$ subslots, it must be that no two edge vectors overlap in the same slot-subslot combination. 

Thus, for each vertex vector there are exactly $(k-1)$ edge vectors which share distinct populated subslots with it, and these edge vectors must cancel out the corresponding subslots i.e., have the same pattern in the shared subslot as that of the vertex vector. In other words, for any two vertex vectors corresponding to slots $i$ and $j$ respectively ($i < j$), the edge vector corresponding to the pair $(i,j)$ must cancel one subslot from each one of the two vertex vectors. This is possible only if (i) the $k$ vertex vectors correspond to distinct vertices in $G$, and (ii) each pair of these vertices have an edge between them for the corresponding edge vector to exist. This implies that $G$ has a $k$-clique which is a contradiction.

%% file: evenset.tex
\section{Parameterized Reduction for the $k$-\textsc{EvenSet} problem} \label{sec:even}

The following is a restatement of Theorem \ref{thm:even}.

\begin{theorem}[Hardness of $k$-\textsc{EvenSet}]
	
	Given an instance $({\bf M},{\bf t})$ of $k$-\textsc{VectorSum}, where ${\bf M} \in {\F}^{m \times n}_2$ and ${\bf t} \in {\F}^m_2$, there is a $\tn{poly}(m,n)$ time reduction to an instance ${\bf M}^\prime$ of $O(k^2\log^2n)$-\textsc{EvenSet}, where ${\bf M}^\prime \in {\F}^{m^\prime \times n^\prime}_2$ such that $m^\prime$ and $n^\prime$ are polynomial  in $n$ and $m$.
\end{theorem}
The rest of this section is devoted to proving the above theorem. The next few paragraphs give an informal description  of the reduction. We then define the variables and equations of the \keven instance, and analyze the completeness and soundness of the reduction.
	
	\subsection{Reduction Overview}
	Let ${\bf M}{\bf x} = {\bf t}$ be a hard instance of \kvec i.e., in the \tn{YES} case there exists a $k$-sparse solution, whereas in the \tn{NO} case all solutions have Hamming weight at least $(k+1)$. We homogenize this affine system by replacing the target vector ${\bf t}$ by $a_0{\bf t}$ for some $\F_2$-variable $a_0$, where $a_0{\bf t}$ is a coordinate-wise multiplication of ${\bf t}$ with the scalar $a_0$. Clearly, if all $(k+1)$-sparse (including $a_0$ as a variable) solutions to ${\bf M}{\bf x} = a_0{\bf t}$ have $a_0 = 1$ then the hardness of \kvec implies the desired hardness result for \keven. However, this may not be true in general: there could exist a $k$-sparse ${\bf x}$ such that ${\bf M}{\bf x} = {\bf 0}$. 
	The objective of our reduction therefore, is to ensure that any solution to  ${\bf M}{\bf x} = a_0{\bf t}$ that has $a_0 = 0$ with a $k$-sparse ${\bf x}$, must have significantly large weight in other auxiliary variables which we shall add in the construction. 
	
	Towards this end, we borrow some techniques from the proof of the inapproximability of \textsc{MinimumDistance} by Austrin and Khot~\cite{AK14}. Using transformations by suitable codes we first obtain a $K = O(k\log n)$-length \emph{sketch} ${\bf y} = (y_1, \dots, y_K)$ of ${\bf x}$, such that ${\bf y}$ is of normalized weight nearly $1/2$ when ${\bf x}$ is $k$-sparse but nonzero. 
	We then construct a codeword ${\bf Y} \in {\F}^{K \times K}_2$, which is intended to  be the product codeword ${\bf y}{\bf y}^T$. However, this relationship cannot be expressed explicitly in terms of linear equations. Instead, for each pair of coordinates $(i,j) \in [K]\times[K]$, we introduce functions $Z_{ij}:{\F}_2\times{\F}_2 \mapsto {\F}_2$ indicating the value taken by the pair $(y_i,y_j)$ along with constraints that relate the $Z_{ij}$ variables to codewords ${\bf y}$ and ${\bf Y}$. In fact, the explicit variables $\{Z_{ij}\}$ determine both ${\bf y}$ and ${\bf Y}$ which are implicit. The constraints also satisfy the key property: if ${\bf x}$ is $k$-sparse, then the number of nonzero $Z_{ij}$ variables is significantly larger when $a_0 = 0$ than when $a_0 = 1$. This forces all sparse solutions to set $a_0 = 1$, which gives us the desired separation in sparsities between the \tn{YES} and \tn{NO} cases.

	\subsection{Constraints}
		
	Let ${\bf M}{\bf x} = {\bf t}$ be the instance of \textsc{$k$-VectorSum} over ${\F}_2$, in $n$ variables and $m$ equations. We homogenize this system  by introducing a new $\F_2$-variable $a_0$ so that the new system of equations is then given by 
	\begin{equation}			\label{eqn:tvs}
		{\bf M}{\bf x} = a_0{\bf t},
	\end{equation} 
	where the $a_0{\bf t}$ is the coordinate wise product of ${\bf t}$ with the scalar $a_0$. We also add the following additional constraints and variables.

	
	
	\noindent\emph{Linear Sketch Constraints} : Let ${\bf R} \in {\F}^{k^{\prime} \times n}$ be the parity check matrix of a $[n, n - k^\prime, 18k ]$ linear code, where $k^{\prime} =  20k\log n$, as defined in Corollary \ref{eq:bch_dist}. Define $\bm{\eta}$ to be a $k^\prime$-length sketch of ${\bf x}$ using ${\bf R}$ as,
	\begin{equation}				\label{eqn:sketch}
		\bm{\eta} = {\bf R}{\bf x}.
	\end{equation}
	
\noindent\emph{Mixing Constraints} : Let ${\bf C} \in {\F}^{K \times k^\prime}_2$ be the generator matrix of a linear code $\mathcal{C} \subseteq {\F}^K_2$ as defined in \cref{thm:balanced} where $\mathcal{C}$ has relative distance $\frac{1}{2} - \epsilon$  and rate $\Omega({\epsilon^3})$ for some small $\epsilon > 0$ and $K = \frac{k^\prime}{\Omega(\epsilon^3) } \leq \frac{20k\log n}{c\epsilon^3}$, for some constant $c > 0$. We add the constraint	\begin{equation}				\label{eqn:mix}
		{\bf y} = {\bf C}\bm{\eta} = {\bf C}{\bf R}{\bf x}.
	\end{equation}
	
	\noindent\emph{Product Code Constraints} : Let $\mathcal{C}^{\otimes 2}\, := \, \mathcal{C} \bigotimes \mathcal{C}$ be the product code with relative distance $\left(\frac{1}{2} - \epsilon\right)^2$, constructed from $\mathcal{C}$. Let ${\bf Y} = \{Y_{ij}\}_{1\leq i,j\leq K} \in {\F}^{K \times K}_2$ be such that ${\bf Y} = {\bf y}{\bf y}^{\sf{T}}$. To represent this relation linearly, we introduce variables $\{Z_{ij}(a,b)\}_{a,b\in\F_2}$ for each $1\leq i,j\leq K$, which are intended to indicate the value assigned to the pair $(y_i,y_j)$ i.e., $Z_{ij}(a,b) = \mathbbm{1}\{y_i = a, y_j = b\}$. For each $(i,j) \in [K]\times[K]$ we add the following equations,
	\begin{eqnarray}            
	Z_{ij}(0,0) + Z_{ij}(0,1) + Z_{ij}(1,0) + Z_{ij}(1,1) &=& a_0    \label{eqn:const11}\\   
	Z_{ij}(1,0) + Z_{ij}(1,1) &=& y_i        \\      
	Z_{ij}(0,1) + Z_{ij}(1,1) &=& y_j        \\      
	Z_{ij}(1,1) &=& Y_{ij}.              \label{eqn:const12}
	\end{eqnarray}
	Furthermore, we add the constraints
	\begin{equation}			\label{eqn:parity}
	{\bf Q}{\bf Y} = {\bf 0},
	\end{equation}	
 	where ${\bf Q}$ is the parity check matrix for the product code $\mathcal{C}^{\otimes 2}$, and 
	\begin{eqnarray}            \label{eqn:const2}
	Y_{ij} &=& Y_{ji} \qquad\qquad \forall i \ne j, \\
	Y_{ii} &=& y_i \qquad\qquad \forall i \in [K],
	\end{eqnarray}
	so that ${\bf Y}$ preserves the diagonal entries and symmetry of ${\bf y}{\bf y}^T$. Finally, we introduce ${\bf x}^1,{\bf x}^2,\ldots,{\bf x}^{r-1}$ and constraints  
	\begin{equation}				\label{eqn:copies}
	{\bf x}^i = {\bf x} \qquad\qquad \forall i \in [r-1],
	\end{equation}
	where $r =  \frac{K^2}{16} \leq \frac{25k(\log n)^2}{c^2\epsilon^6}$. These $r-1$ explicit copies of the vector ${\bf x}$ are used to balance the Hamming weight of the final solution. Observe that all the variables described above are linear combinations of $a_0, \{Z_{ij}(\cdot,\cdot)\}_{i,j \in [k]}$ and the coordinates of the vectors ${\bf x}$ and $\{{\bf x}^i\}_{i \in [r-1]}$. Hence, we analyze the sparsity of the solution restricted to these explicit variables. The total number of variables considered is $4K^2 + r\cdot n + 1$.
	
	\medskip
	\noindent
	\emph{Remark}: The key difference between \cite{AK14} and our reduction is in Equation (\ref{eqn:sketch}) which constructs a small ($O(k\log n)$)-length sketch of the $n$-length vector ${\bf x}$. This helps us contain the blowup in the sparsity of the solution to $O(k^2\log^2n)$ instead of $O(n)$.

	\subsection{Completeness}
	
	In the \tn{YES} case, setting $a_0 = 1$ we obtain a $k$-sparse ${\bf x}$ such that ${\bf M}{\bf x} = a_0{\bf t} = {\bf t}$. Furthermore, for each $i,j \in [K]$, exactly one of the $Z_{ij}$ variables would be nonzero. Hence, we have a solution of weight $K^2+ r k+1$. 
	
	\subsection{Soundness}
	
	Since the solution has to be non-trivial, at least one of $a_0,{\bf x},{\bf y},{\bf Y}$ must be nonzero. Note that when ${\bf x} = {\bf 0}$, ${\bf y} = {\bf 0}$ since ${\bf y}$ is a homogeneous linear transformation of ${\bf x}$. Moreover, we may assume that the weight of ${\bf x}$ is at most $\frac{K^2 + 1}{r} + k +1 < 18k$ by our setting of $r$, otherwise the total weight of the solution would be at least $r\cdot\left(\frac{K^2 + 1}{r} + k +1\right) \geq K^2+ r(k+1) + 1$ due to the copies of ${\bf x}$ and we would be done. The construction of ${\bf y}$ along with the upper bound of $18k$ on the weight of ${\bf x}$ constrains ${\bf y}$ to be nonzero when ${\bf x}$ is nonzero. Thus, the only three cases we need to consider are:
	
	\medskip
	\noindent
	{\bf Case (i):} $a_0 = 1$. In this case, any solution ${\bf x}$ to ${\bf M}{\bf x} = a_0{\bf t} = {\bf t}$ has weight at least $k+1$. Furthermore, for each $i,j \in [K]$, at least one of the four $Z_{ij} $ variables must be nonzero since $a_0 = 1$. Hence, the total Hamming weight of the solution is at least $K^2 + r(k+1) + 1$.

	\medskip
	\noindent
	{\bf Case (ii):} $a_0 = 0, {\bf x} \neq {\bf 0}, {\bf y}\neq {\bf 0}$. By construction, since ${\bf y}$ is nonzero it has weight $\ge \left(\frac{1}{2}- \epsilon\right)K$. Therefore, for at least $1 - \left(\frac{1}{2} + \epsilon \right)^2 \ge \frac{3}{4} - 2\epsilon$ fraction of the pairs $(i,j) \in [K] \times [K]$, either $y_i = 1$ or $y_j = 1 $ .  Observe that for each such pair, at least two $Z_{ij}$ variables are set to $1$. Thus, the weight of any  solution in this case is at least $2 \Big(\frac{3}{4} - 2 \epsilon \Big)K^2 = \Big(\frac{3}{2} -4 \epsilon \Big)K^2$.

	\medskip
	\noindent
	{\bf Case (iii):} $a_0 = 0, {\bf x} = {\bf 0}, {\bf y} = {\bf 0}, {\bf Y} \neq {\bf 0}$. We have that ${\rm diag}({\bf Y}) = {\bf y} =  {\bf 0}$, ${\bf Y}$ is symmetric and it belongs to the product code $\mathcal{C}^{\otimes 2}$ (as enforced by Equations \eqref{eqn:parity} and \eqref{eqn:const2}). Then by lemma \ref{thm:prod_codes}, the weight of ${\bf Y}$ is at least $\left( \frac{3}{8} -
	3\epsilon \right)K^2$. Observe that for each $i,j \in [K]$ such that $Y_{ij} =1 $, Equations (\ref{eqn:const11})-(\ref{eqn:const12}) force all four $Z_{ij}$ variables to be set to $1$. Hence, the number of nonzero $Z_{ij}$'s are at least $\left(\frac{3}{2} - 12 \epsilon \right) K^2$. 
	
	\medskip
	The above analysis yields that in contrast to  the \tn{YES} case which admits a $(K^2 + rk + 1)$-sparse solution, in the \tn{NO} case all solutions are of weight at least
	$$ \min\left\{\left(K^2 + r(k+1) + 1\right) , \left(\frac{3}{2}- 12 \epsilon\right)K^2\Big)\right\} \geq K^2 + r(k+1) + 1 $$
		by choice of the parameter $r$. Thus, solving the $d$-{\sc EvenSet} problem with 
	$d = K^2 + rk + 1 = O(k^2(\log n)^2)$
	solves the \kvec instance ${\bf M}{\bf x} = {\bf t}$.
	

%% file: k-parity-stdinc.tex
\section{Hardness of Learning $k$-Parities using $k$-Juntas} \label{sec:kparityjunta}
The hardness for \kVS proved in Theorem \ref{thm:vechard} can be restated
in terms of W$[1]$-hardness of learning $k$-parities, i.e. linear forms
depending on at most $k$-variables.

\begin{theorem}\label{thm-kparity} The following is W$[1]$-hard: given
$r = O(k\log n)$ point-value pairs $\{({\bf y}_i, a_i)\}_{i=1}^r
\subseteq \Ft^n\times \Ft$,
decide whether there exists a homogeneous 
linear form $L$ supported on at most $k$ variables 
which satisfies all the point-value pairs, i.e. $L({\bf y}_i) = a_i$
for all $i = 1,\dots, t$.
\end{theorem}

Combining the above with a small bias linear code we 
induce an approximation gap for learning $k$-parities 
along with extending the result to non-homogeneous linear forms.

\begin{theorem}\label{thm-kparityapprox} The following is W$[1]$-hard:
for any $\eps > 0$ depending only on $k$, given
$t = O(k\log n/\eps^3)$ point-value pairs $\{({\bf z}_i,
b_i)\}_{i=1}^t
\subseteq \Ft^n\times \Ft$,
decide whether: 

\smallskip
\noindent
\tn{YES Case:} There exists a homogeneous 
linear form supported on at most $k$ variables
which satisfies all the point-value pairs.

\smallskip
\noindent
\tn{NO Case.} Any linear form supported on at most $k$
variables satisfies a fraction in the range $[1/2 - \eps, 1/2 + \eps]$ 
of the point value pairs.
\end{theorem}
\begin{proof} Let ${\bf W} = \{W_{ij}\} \in \Ft^{t \times r}$ be the generator
	matrix of an $\eps$-balanced linear code given by Theorem \ref{thm:balanced},
	where $t = O(r/\eps^3)$.
Given an instance $\{({\bf y}_j, a_j)\}_{j=1}^r$ from Theorem \ref{thm-kparity}, let
$${\bf z}_i = \sum_{j=1}^r {W}_{ij}{\bf y}_j, \ \ \ \ \tn{ and }
\ \ \ \ b_i = \sum_{j=1}^r {W}_{ij}a_j,$$
for $i = 1, \dots, t$.

In the YES case, there is a homogeneous linear form $L^*$ that
satisfies all $\{({\bf y}_j, a_j)\}_{j=1}^r$ and thus satisfies
linear combinations of these point-value pairs, in particular 
$\{({\bf z}_i, b_i)\}_{i=1}^t$. 

For the NO case, consider any linear form $L({\bf x}) + c$. Since the
homogeneous part $L$ does not satisfy all pairs $\{({\bf y}_j,
a_j)\}_{j=1}^r$, it will
satisfy a fraction in the range $[1/2 - \eps, 1/2 + \eps]$ of the
pairs $\{({\bf z}_i,
b_i)\}_{i=1}^t$, due the lower and upper bounds bound on the weight of
the nonzero codewords in the column space of ${\bf W}$. Thus, the
linear form $L({\bf x}) + c$ also satisfies a fraction in the range
$[1/2 - \eps, 1/2 + \eps]$ of the point-value pairs
$\{({\bf z}_i, b_i)\}_{i=1}^t$.
\end{proof}

As we show below, using a small enough bias $\eps$ in the above
construction, one can strengthen the hardness result to 
learning $k$-parities with $k$-\emph{juntas}, i.e. functions
depending only on a subset of at most $k$ variables.

\begin{theorem}\label{thm-kparityjuntarestated} \tn{(Theorem \ref{thm-kparityjunta} restated)}
The following is W$[1]$-hard:
for any constant $\delta > 0$, given
$t = O(k\cdot 2^{3k}\cdot \log n/\delta^3)$ point-value pairs $\{({\bf z}_i,
b_i)\}_{i=1}^t
\subseteq \Ft^n\times \Ft$,
decide whether: 

\smallskip
\noindent
\tn{YES Case:} There exists a homogeneous 
linear form supported on at most $k$ variables
which satisfies all the point-value pairs.

\smallskip
\noindent
\tn{NO Case.} Any function $f : \Ft^n \mapsto \Ft$ depending on at most $k$
variables satisfies at most $1/2 + \delta$ fraction of the point value
pairs.
\end{theorem}
\begin{proof}
The construction of $\mc{Z} = 
\{({\bf z}_i, b_i)\}_{i=1}^t$ is exactly the same
as in the proof of Theorem \ref{thm-kparityapprox} taking $\eps =
\delta\cdot 2^{-k}$. The YES case follows directly as before.

For the NO case, let $f: \Ft^n \mapsto \Ft$ be a function depending
only a subset $S \subseteq [n]$ of coordinates where $|S| \leq k$. 
Define an extension $g : \Ft^{n+1} \mapsto
\Ft$ as $g(x_1,\dots, x_n, x_{n+1}) := f(x_1,\dots, x_n) + x_{n+1}$.
For convenience we shall abuse notation to denote 
$({\bf z}, b) = (z_1,\dots, z_n, b)$ where ${\bf z} = (z_1,\dots,
z_n) \in \Ft^n$ and $b \in \Ft$. To complete the proof we need to show that,
\begin{equation}
\left|\E_{({\bf z}, b) \in \mc{Z}}\left[e(g({\bf z}, b))\right]\right| 
\leq 2\delta, \label{eqn-expbd}
\end{equation}
where $e(x) := (-1)^x$. For some real values
$C_\alpha$ ($\alpha \subseteq [n+1]$), the Fourier expansion of $e(g)$
is given by,
$$e(g) = \sum_{\alpha \subseteq [n+1]} C_\alpha\chi_\alpha.$$
Since $e(g(x_1,\dots, x_{n+1})) = e(f(x_1,\dots, x_n) + x_{n+1})$ and
$f$ depends only on coordinates in $S$, it
is easy to see that the Fourier spectrum of $e(g)$ is supported only
on characters $\chi_\alpha$ such that $\alpha \subseteq S\cup\{n+1\}$. 
Further, since $e(g(x_1,\dots, x_{n+1}))$ changes sign on flipping
$x_{n+1}$, $C_\alpha \neq 0 \Rightarrow (n+1)\in \alpha$. Thus,
\begin{equation}
e(g) = \sum_{\substack{\alpha \subseteq S\cup\{n+1\} \\ (n+1)\in
\alpha}} C_\alpha\chi_\alpha. \label{eqn-Fourier}
\end{equation}
Observe that for any $\alpha$ in the sum above,
$\chi_{\alpha}(x_1,\dots,x_n,b) = e(L(x_1,\dots,x_n) + b)$ where $L$
is a homogeneous linear form supported on at most $k$ variables. For
any such $\alpha$,
the NO case of Theorem  \ref{thm-kparityapprox} implies,
\begin{equation}
\left|\E_{({\bf z}, b) \in \mc{Z}}\left[\chi_\alpha({\bf z}, b) \right]\right| \leq
2\eps
\end{equation}
Using the above along with Equation \eqref{eqn-Fourier} yields,
\begin{eqnarray}
\left|\E_{({\bf z}, b) \in \mc{Z}}\left[e(g({\bf z}, b))\right]\right|
& \leq & (2\eps)\cdot \sum_{\substack{\alpha \subseteq S\cup\{n+1\} \\ (n+1)\in
\alpha}} |C_\alpha| \nonumber \\
& \leq & (2\eps)\cdot 2^k = 2\delta, \nonumber
\end{eqnarray}
where the last inequality is because there are at most $2^k$ subsets
$\alpha$ in the sum on the RHS of Equation \eqref{eqn-Fourier} and each
$|C_\alpha| \leq 1$ since $e(g)$ is a $\{-1,1\}$-valued function.
\end{proof}

%% file: evenset-polynomials.tex
\section{Proof of Theorem \ref{thm:kevenwithpolys}} \label{sec:kevenwithpolys}
We first prove the following strengthening of Theorem \ref{thm:even}
along the same lines as Theorem
\ref{thm-kparityapprox} in Section \ref{sec:kparityjunta}. 

\begin{theorem}[Hardness of approximate \keven]\label{thm:evenpoly} 
For any constant $\eps > 0$, given an instance
$({\bf A}, {\bf b})$ of $k'${\sc -VectorSum}, 
where ${\bf A} \in \F_2^{m' \times n'}$
and ${\bf b} \in \F_2^{m'}$, there is an \tn{FPT} reduction to an
instance ${\bf B} \in \F_2^{m\times n}$ of 
$k'$-\textsc{EvenSet} for some $k = O((k'\log
{n'})^2)$, such that

\smallskip
\noindent
\tn{YES Case:} There is a nonzero $k$-sparse vector 
${\bf x}$ which satisfies ${\bf Bx} = 0$.

\smallskip
\noindent
\tn{NO Case.} For any nonzero $k$-sparse vector ${\bf x}$
the weight of ${\bf Bx}$ is in the range $[1/2 - \eps, 1/2 + \eps]$.

\smallskip
\noindent
Here both $m$ and $n$ are bounded
by fixed polynomials in $m'$ and $n'$. 
\end{theorem}
\begin{proof} 
Let ${\bf M} \in \F_2^{r\times n}$ be the instance of \keven obtained
by applying Theorem \ref{thm:even} to the instance $({\bf A}, {\bf
b})$ of $k'${\sc -VectorSum} we start with.
As in the proof of Theorem
\ref{thm-kparityapprox} let ${\bf W} \in \F_2^{m \times r}$ be the generator
matrix of an $\eps$-balanced linear code given by Theorem \ref{thm:balanced}, where $m =
O(r/\eps^3)$. Taking ${\bf B} := {\bf WM}$ completes the proof.\end{proof}

It is easy to see that the uniform distribution on the rows of the
matrix ${\bf B}$ fools all linear forms (with error $\eps$) over $k$ variables. Viola's
result~\cite{Viola09} (Theorem \ref{thm:viola}) implies that for any
constant $d$, taking $d$-wise sums of the rows of ${\bf M}$ yields a
distribution which fools all degree $d$ polynomials with error
$16\cdot \eps^{1/2^{d-1}}$. Taking $\eps$ to be a small enough constant yields
the following theorem which implies Theorem \ref{thm:kevenwithpolys}.

\begin{theorem}\label{thm:kevenwithpolsredn}
For any constants $\delta > 0$ and positive integer $d$, 
given an instance $({\bf A}, {\bf b})$ of $k'${\sc -VectorSum}, 
where ${\bf A} \in \F_2^{m' \times n'}$
and ${\bf b} \in \F_2^{m'}$, there is an \tn{FPT} reduction to a
set of  $m$ points $\{{\bf z}_i\}_{i=1}^m
\subseteq \F_2^n$ such that for some $k = O((k'\log n')^2)$,

\smallskip
\noindent
\tn{YES Case:} There exists a $k$-parity $L$ such that $L({\bf z}_i) = 0$ for all $i=1,\dots, m$.

\smallskip
\noindent
\tn{NO Case.} Any degree $d$ polynomial $P : \F_2^n \mapsto \F_2$ depending on at most $k$ variables satisfies $P({\bf z}_i) = 0$ for
 at most $\left(\Pr_{{\bf z} \in \F_2^n}[P({\bf z}) = 0] + \delta\right)$ fraction of the points. 

\smallskip
\noindent
In the above $m$ and $n$ are bounded by fixed polynomials in $m'$ and $n'$. 
\end{theorem}

%% file: Gap-k-parity-stdinc.tex
\section{Hardness Reduction for \GkVS} \label{sec-gapink}

In this section we first prove the following theorem.
\begin{theorem}\label{thm-Gap-kVS-redn} For universal constants
$\delta_0 > 0$ and $c_0$, and any arbitrarily small constant $\eps > 0$,
there is a $2^{O(\eps n)}$-time Turing reduction from an $n$-variable 
$3$\tn{-CNF} formula
$F$ to $s \leq 2^{\eps n}$ instances $I_1, \dots, I_s$ of \GkVS,
each of size at most $O(k^2\cdot 2^{O(n'/k)})$ where $n' = c_\eps n(\log
n)^{c_0}$ for some constant $c_\eps$ depending on $\eps$, such that

\smallskip
\noindent
\tn{YES Case.} If $F$ is satisfiable, then at least one instance $I_j$
($j \in [s]$) admits a solution of sparsity $k$.

\smallskip
\noindent
\tn{NO Case.} If $F$ is unsatisfiabe, then no instance $I_j$ 
admits a solution of sparsity $\leq (1 + \delta_0)k$.

\end{theorem}
\begin{proof}
Let $F$ be a $3$-CNF formula on $n$-variables. We use Lemma
\ref{lem-sparsification} to obtain $3$-CNF fomulas $H_1, \dots H_s$ for $s
\leq 2^{\eps n}$, each on at most $n$ variables and 
$O((1/\eps)^9 \cdot n)$ clauses. Using Theorem
\ref{thm-nearlinearPCP},
each $H_j$ ($j \in [s]$) 
is separately transformed into a Gap-$3$-SAT instance $F_j$
with at most $n' = c_\eps n(\log n)^{c_0}$ clauses (and variables),
for some $c_\eps = O((1/\eps)^9)$. Each $F_j$ is now reduced to an
instance $I_j$ of \GkVS as follows.

Fix $j \in [s]$ and let $\mc{C}(F_j)$ denote the set of clauses of $F_j$.
We may assume that $|\mc{C}(F_j)|$ is divisible by $k$ by adding up to
$k$ dummy clauses which are always satisfied.
Let $B^j_1,\dots, B^j_k$ be any arbitrary partition of $\mc{C}(F_j)$ into $k$
equal sized subsets, and let $X(B^j_i)$ be the set of variables
contained in the clauses $B^j_i$. Let $\tn{SAT}(B^j_i) \subseteq
\{0,1\}^{X(B^j_i)}$ be the set of assignments to the variables
$X(B^j_i)$ that satisfy all the clauses in $B^j_i$, for $1\leq i\leq
k$. For each $\alpha \in \tn{SAT}(B^j_i)$ we introduce an $\Ft$-valued
variable $z_\alpha$, and add the following equation,
\begin{equation}
\sum_{\alpha \in \tn{SAT}(B^j_i)} z_\alpha = 1, \label{eqn-blocksum}
\end{equation}
for each $i = 1,\dots, k$.
We also add equations to ensure consistency of the solution 
across different blocks of
variables. For each pair of distinct blocks $B^j_i$ and $B^j_{i'}$ and each
assignment to their common variables $\sigma \in \{0,1\}^{X(B^j_i)\cap
X(B^j_{i'})}$ we add the following equation.
\begin{equation}\label{eqn-blockconsistency}
\sum_{\substack{\alpha \in \tn{SAT}(B^j_i) \\ \sigma = \alpha|_{X(B^j_i)\cap
X(B^j_{i'})} }} z_\alpha = \sum_{\substack{\beta \in
\tn{SAT}(B^j_{i'}) \\ \sigma = \beta|_{X(B^j_i)\cap X(B^j_{i'})}}}
z_\beta.
\end{equation}
It is easy to see that $|X(B^j_i)| = O(n'/k)$ and
$|\tn{SAT}(B^j_i)| \leq 2^{O(n'/k)}$ for $1 \leq i \leq k$. 
Thus, the above construction of $I_j$ has at most $k\cdot 2^{O(n'/k)}$
variables, and at most $k^2\cdot 2^{O(n'/k)}$ equations. The
constant $c_0$ is the same as in Theorem \ref{thm-nearlinearPCP} and
$\delta_0$ shall be chosen later.

\subsection{YES Case}
If $F$ is satisfiable, the for some $j^* \in [s]$, $H_{j^*}$ is
satisfiable and therefore $F_{j^*}$ is satisfiable. Let $\pi$ be a
satisfying assignment for $F_{j^*}$. In $I_{j^*}$, for each
$B^{j^*}_i$ ($1\leq i \leq k$) we set the variable $z_\alpha$
corresponding to the projection $\alpha \in\tn{SAT}(B^{j^*}_i)$  
of $\pi$ on to the variables
$X(B^{j^*}_i)$ to $1$, and the rest of the variables to $0$. Clearly,
this satisfies Equations \eqref{eqn-blocksum} and
\eqref{eqn-blockconsistency} since $\pi$ is a satisfying assignment.
As we set exactly one variable per block to $1$, 
$I_{j^*}$ admits a solution of sparsity $k$.

\subsection{NO Case}
If $F$ is not satisfiable, then none of $H_1,\dots, H_s$ are
satisfiable and thus, for each $j = 1, \dots, s$, at most $(1 -
\gamma_0)$ fraction of the clauses of $F_j$ are satisfiable by any
assignment. Fix any $j \in [s]$. Since $B^j_1, \dots, B^j_k$ is a
balanced partition of $\mc{C}(F_j)$, any assignment to the variables of $F_j$
can satisfy all the clauses of at most $(1 - \gamma_0)$ fraction of
 $B^j_i$, $1\leq i\leq k$. 

Consider a solution to $I_j$ of sparsity at most $(1 + \delta_0)k$. 
By Equation
\eqref{eqn-blocksum}, this solution must set
an odd number of variables in $\{z_\alpha \mid \alpha \in
\tn{SAT}(B^j_i)\}$ to $1$, for $1\leq i \leq k$. Let $S \subseteq [k]$
consist of all indices $i \in [k]$ such that exactly one variable
$z_{\alpha^i}$ for some $\alpha^i \in \tn{SAT}(B^j_i)$ 
is set to $1$. Thus,
the sparsity is at least $|S| + 3(k - |S|)$, which is at most $(1 +
\delta_0)k$ by our assumption. By rearranging we get $|S| \geq (1 - \delta_0/2)k$. 
Further, Equation \eqref{eqn-blockconsistency} implies that 
for any $i, i' \in S$, the
assignments $\alpha^i$ and $\alpha^{i'}$ are consistent on their
common variables. Thus, there is an assignment to the variables in
$\cup_{i\in S}X(B^j_i)$ that satisfies all the clauses of $B^j_i$ for
$i\in S$. Choosing $\delta_0 = \gamma_0$ yields a contradiction to our
assumption, and therefore
no $I_j$ admits a solution of sparsity at most $(1 + \delta_0)k$. 
\end{proof}

Theorem \ref{thm-Gap-kVS-redn} proved above 
implies the following restatement of Theorem \ref{thm-gapink}.

\begin{theorem}\label{thm-rulingoutrestated} Assuming the Exponential Time
Hypothesis, there are universal constants $\delta_0 >0$ and $c_0$ 
such that  there is no $\tn{poly}(N)$ time algorithm to determine whether an
instance of \GkVS of  size $N$ admits a solution of sparsity $k$ or all solutions
are of sparsity at least $(1 + \delta_0)k$, 
for any $k = \omega((\log\log N)^{c_0})$.
More generally, under the same assumption, this problem does not admit an 
$N^{O(k/\omega((\log\log N)^{c_0}))}$ time algorithm for unrestricted $k$. 
\end{theorem}
\begin{proof}
For the first part of the theorem, assume 
for a contradiction that such an algorithm exists. In Theorem
\ref{thm-Gap-kVS-redn}, the size of each \GkVS instance constructed is
at most $N$, where $\log N = O(\log k + n'/k) = O(\log k) + O(c_\eps n(\log
n)^{c_0}/k)$. Here $n$ is the number of variables in the $3$-CNF
formula $F$. Thus, choosing $k = \omega((\log\log N)^{c_0})$ implies
$k = \omega((\log n)^{c_0})$. Note that in the reduction $k$ is
bounded by $\tn{poly}(n)$. Our supposed algorithm would decide each
\GkVS instance in time $\tn{poly}(k,  2^{O(n'/k)}) = 2^{o(n)}$. Applying this to all the
instances of \GkVS would decide the $n$-variable $3$-CNF formula $F$ in time $2^{(\eps +
o(1))n}$ for all constants $\eps > 0$, which contradicts the ETH.

A similar analysis proves the second part (unrestricted $k$) of the theorem.
\end{proof}

%% file: k-parity-algo.tex
\section{A simple $O(n\cdot 2^m)$-time algorithm for \kVS}\label{sec-kparityalgo}
Let $({\bf M}, {\bf b})$ be an instance of \kVS where ${\bf M} 
\in \Ft^{m\times n}$ and ${\bf b} \in \Ft^{m}$. Construct a graph $G$ on
vertex set $V = \Ft^m$ and edge set given by,
$$E = \left\{\{{\bf u}, {\bf v}\} \in {V \choose 2} \mid {\bf u} + {\bf v}
\tn{ is a column of } {\bf M}\right\}.$$
We say that an edge $\{{\bf u}, {\bf v}\} \in E$ is labeled by the
column ${\bf u} + {\bf v}$ of ${\bf M}$.
Clearly, if there is a vector ${\bf x}$ of Hamming weight at most $k$ such
that ${\bf Mx} = {\bf b}$ then there is a path of length at most $k$
in $G$ from ${\bf 0}$ to ${\bf b}$ given by choosing the edges labeled
by the columns corresponding to the non-zero entries of ${\bf x}$ in
any sequence. On the other hand, if there is a path in $G$ from ${\bf
0}$ to ${\bf b}$ of length at most $k$, then there is a sequence of at
most $k$ columns (with possible repetitions) of ${\bf M}$ which sum
up to ${\bf b}$. Cancelling out even number of repetitions of any
column yields a subset of at most $k$ distinct columns of ${\bf M}$
that sum up to ${\bf b}$. Thus, deciding \kVS reduces to determining
whether there is a path of length at most $k$ from ${\bf 0}$ to ${\bf
b}$.  

The size of $V$ is $2^r$ and of $E$ is at most $n\cdot2^m$, and the graph
can be constructed in time $O(n\cdot 2^m)$. 
Doing a Breadth First Search yields a running time of 
$O(n\cdot 2^m)$.

%% file: Parity-stdinc.tex
\section{Optimal hardness of learning parities using degree $d$
polynomials} \label{sec-Khot-personal}

The hardness reduction in this section is due to Khot~\cite{Khot-personal}.

The starting point of the reduction is the {\sc MinimumDistance} problem over $\F_2$: given a matrix ${\bf A} \in \F_2^{m\times n}$, find
a nonzero vector ${\bf z} \in \F_2^n$ to minimize $\overline{\tn{wt}}(\textbf{Az})$
where $\overline{\tn{wt}}({\bf a})$ is the normalized hamming weight of
${\bf a}$. The latter quantity denotes the distance of the code given
by the linear span of the columns of ${\bf A}$.

Below we restate the hardness of {\sc MinimumDistance} as proved by Austrin and
Khot~\cite{AK14} along with an additional guarantee satisfied by their
reduction in the YES case which shall prove useful for the subsequent
application.

\begin{theorem}\tn{[\cite{AK14}]} There is a universal constant $\zeta \in (0, 1/5)$
such that given a matrix ${\bf A} \in \F_2^{m\times n}$, it is NP-hard to
distinguish between the following:

\smallskip
\noindent
\tn{YES Case.} There is a vector ${\bf z} = (z_1, \dots, z_n) \in
\F_2^n$ such that $z_1 = 1$ and $\overline{\tn{wt}}({\bf Az}) \leq \zeta$.

\smallskip
\noindent
\tn{NO Case.} For any nonzero vector ${\bf z} \in \F_2^n$,
$\overline{\tn{wt}}({\bf Az}) \geq 5\zeta$. 

\end{theorem}

The following is an easy consequence of above obtained by tensoring
the instance of the above theorem.

\begin{theorem}\label{thm-tensored-MDC} 
There is a universal constant $\zeta \in (0, 1/5)$
such that for any constant positive integer $K$, 
given a matrix ${\bf A} \in \F_2^{m\times n}$, it is NP-hard to
distinguish between the following:

\smallskip
\noindent
\tn{YES Case.} There is a vector ${\bf z} = (z_1, \dots, z_n) \in
\F_2^n$ such that $z_1 = 1$ and $\overline{\tn{wt}}({\bf Az}) \leq \zeta^K$.

\smallskip
\noindent
\tn{NO Case.} For any nonzero vector ${\bf z} \in \F_2^n$,
$\overline{\tn{wt}}({\bf Az}) \geq 5^K\zeta^K$. 

\end{theorem}

From the above we obtain -- using techniques similar to those
used in \cite{Khot06} --  the following hardness of distinguishing
between an {\sc MinimumDistance} instance of distance $\eps$ vs. $1/2 - \eps$. 

\begin{theorem} \label{thm-MDC-main}
For any positive constant $\eps > 0$,
given a matrix ${\bf B} \in \F_2^{m\times n}$, it is NP-hard to
distinguish between the following:

\smallskip
\noindent
\tn{YES Case.} There is a vector ${\bf z} = (z_1, \dots, z_n) \in
\F_2^n$ such that $z_1 = 1$ and $\overline{\tn{wt}}({\bf Bz}) \leq \eps$.

\smallskip
\noindent
\tn{NO Case.} For any nonzero vector ${\bf z} \in \F_2^n$,
$1/2 - \eps \leq \overline{\tn{wt}}({\bf Bz}) \leq 1/2 + \eps$. 

\end{theorem}
\begin{proof}
Let ${\bf A} \in \F_2^{m\times n}$ be an instance of {\sc MinimumDistance} from Theorem
\ref{thm-tensored-MDC} for a large value of $K$. Let ${\bf A}_i,
i=1,\dots, m$ be the rows of ${\bf A}$. Let $G$ be a regular expander on $m$
vertices labeled from $[m]$, with degree $D = 
\left(\frac{4}{5^K\zeta^K}\right)^{10}$ and second largest
eigenvalue $\lambda \leq D^{0.9}$, which can be constructed
efficiently. Note that $\lambda/D \leq (5^K\zeta^K)/4$.

Let $t = 1/(2^K\zeta^K)$ and consider all $t$ length random walks
$[i_1, i_2, \dots, i_t]$ in $G$. There are $m\cdot D^t$ of such walks.
For each such walk we add $2^t$ rows in ${\bf B}$ given by
$$\left\{\sum_{j=1}^t s_jA_{i_j}\ \mid\ s_1,\dots, s_t \in \F_2\right\}.$$
In total, there are $m' = m\cdot (2D)^t$ rows in ${\bf B}$ and $n$
columns. 

\medskip
\noindent
{\bf YES Case}

\smallskip
\noindent
Let ${\bf z}$ be the vector given in the YES case of Theorem
\ref{thm-tensored-MDC}, such that there are 
at most $\zeta^K$ fraction of rows ${\bf A}_i$ satisfying $\langle {\bf
A}_i, {\bf z}\rangle = 1$. For a random walk in $G$, the probability
that it contains an index corresponding to any such row is at most
$t\zeta \leq 1/2^K$. Thus, $\overline{\tn{wt}}({\bf Bz}) \leq 1/2^K$.

\medskip
\noindent
{\bf NO Case}

\smallskip
\noindent
Consider any ${\bf z} \in \F_2^n$. From the NO case of Theorem
\ref{thm-tensored-MDC}, we have that at least $5^k\zeta^K$
fraction of the rows ${\bf A}_i$ satisfy $\langle {\bf
A}_i, {\bf z}\rangle = 1$. Let $I$ be the set of the indices
corresponding to these rows. Using the
analysis in Section 4.7 of \cite{LW} we can bound the probability that
a $t$ length random walk $[i_1,\dots, i_t]$ 
in $G$ does not contain an index from $I$ as
follows.
\begin{eqnarray}
\Pr\left[i_1 \not\in I, i_2 \not\in I, \dots, i_t \not\in I\right] &
\leq & \left(\sqrt{1 - 5^K\zeta^K} + \frac{\lambda}{D}\right)^t
\nonumber \\
& \leq & \left(\sqrt{1 - 5^K\zeta^K} + \frac{5^K\zeta^K}{4}\right)^t
\nonumber \\
& \leq & \left(1 - \frac{5^K\zeta^K}{4}\right)^{1/(\zeta^K2^K)}
\nonumber \\
& \leq & e^{-(5/2)^K/4} \leq 2^{-2^K}. \nonumber 
\end{eqnarray}
Thus, at least $1 - 2^{-2^K}$ fraction of the random walks contain an
index from $I$. For any such walk, exactly half of the $2^t$ linear
combinations of the corresponding rows result in a row ${\bf B}_r$ of
${\bf B}$ such that $\langle {\bf B}_r, {\bf z}\rangle = 0$. Thus,
$(1/2)(1 - 2^{-2^K}) \leq \overline{\tn{wt}}({\bf Bz}) \leq (1/2)(1 + 2^{-2^K})$. 
Choosing $K$ to be large
enough completes the proof.

\end{proof}

\subsection{Main Theorem}
The following theorem shows optimal hardness of learning linear forms by
degree $d$ polynomials and subsumes the weaker results in \cite{GKS10}.

\begin{theorem} \label{thm-learning-main} For any constants $\delta > 0$
and positive integer $d$, given a set of point-value
pairs $\{({\bf x}_i, y_i)\}_{i=1}^M \subseteq \F_2^M\times \F_2$, it is
NP-hard to distinguish whether

\smallskip
\noindent
\tn{YES Case.} There is a linear form $\ell^* : \F_2^N \mapsto \F_2$
such that for at least $(1 - \delta)$ fraction of the pairs $({\bf x},
y)$, $\ell^*({\bf x}) = y$.

\smallskip
\noindent
\tn{NO Case.} Any degree $d$ polynomial $p : \F_2^N \mapsto \F_2$
satisfies $p({\bf x}) = y$ for at most $(1/2 + \delta)$ fraction of the
point-value pairs.
\end{theorem}
\begin{proof}
Let ${\bf B} \in \F_2^{m\times n}$ be the matrix given by Theorem
\ref{thm-MDC-main} for a parameter $\eps > 0$ to be chosen small
enough. For
each combination of $d$ rows of ${\bf B}$ we add a point-value pair in our
instance as follows. For ${\bf r} = (r_1, \dots, r_n)$ obtained by
summing some $d$ rows of ${\bf B}$ we add $((r_2,\dots, r_n), r_1)$ to
our instance. Thus, $M = m^d$ and $N = n-1$.

\medskip
\noindent
{\bf YES Case}

\smallskip
\noindent
Let ${\bf z} = (z_1, \dots, z_n)$ with $z_1 = 1$ be the vector 
from the YES case of Theorem
\ref{thm-MDC-main} satisfying $\langle {\bf z}, {\bf b}\rangle = 0$
for at least $(1 - \eps)$ fraction of the rows ${\bf b}$ of ${\bf B}$.
Defining $\ell^*({\bf x}) = \sum_{j=1}^{n-1} z_{j+1}x_j$, we obtain
that for $(1 - \eps)$ fraction of the rows $(b_1, \dots, b_n)$ of $B$,
$\ell^*(b_2,\dots, b_n) = b_1$. Thus, for at least $(1 - \eps d)$
fraction of the point-value pairs $({\bf x}, y)$ of the instance
constructed above, $\ell^*({\bf x}) = y$.

\medskip
\noindent
{\bf NO Case}

\smallskip
\noindent
In the NO case, the uniformly random distribution over the 
rows of ${\bf B}$ fools every linear form $\ell : \F_2^n \mapsto \F_2$
with bias $\eps$. By the result of Viola~\cite{Viola09} 
(Theorem \ref{thm:viola}),
the distribution given by uniformly random $d$-wise sums of of the
rows of ${\bf B}$ fools all degree $d$ polynomials $q : \F_2^n \mapsto
\F_2$ with bias $\eps_d = 16\cdot \eps^{1/2^{d - 1}}$. Let $(r_1,
\dots, r_n)$ be a random element from the latter distribution, and let
$q(r_1, \dots, r_n) = r_1 + p(r_2, \dots, r_n)$ where $p$ is a degree $d$
polynomial. Since $q$ is linear in the first coordinate, the bias of
$q$ under the uniform distribution is $0$, and
$p(r_2,\dots, r_n) \neq r_1$ for at least $1/2 - \eps_d$ fraction of
$(r_1, \dots, r_n)$.

Thus, for any degree $d$ polynomial $p$, for at least $1/2 - \eps_d$
fraction of the point-value pairs $({\bf x}, y)$ constructed in our
instance $p({\bf x}) \neq y$.

To complete the reduction we choose $\eps$ to be small enough so that
$\max\{\eps d, 16\cdot \eps^{1/2^{d - 1}}\} \leq \delta$.
 
\end{proof}